%% file: QCCProduct_ZehLing.tex
\documentclass[conference, final, letterpaper]{IEEEtran}
\usepackage[utf8]{inputenc}
\usepackage{graphicx}
\usepackage{stfloats}
\usepackage{amssymb}
\usepackage[cmex10]{amsmath}
\usepackage{mathtools}
\usepackage{amsthm} 
\usepackage{array}
\usepackage[usenames,dvipsnames]{xcolor}
\usepackage{hyperref}
\hypersetup{colorlinks=false, pdfborderstyle={/S/U/W 1},pdfborder=0 0 1, citebordercolor=Blue, filebordercolor=Red, linkbordercolor=Red, urlbordercolor=Blue}
\usepackage{url}
\usepackage{xparse}
\usepackage{subfigure}
\usepackage{nicefrac}
\usepackage[backend=bibtex8, style=ieee, url=false, isbn=false, doi=false, maxnames=5, firstinits=true]{biblatex}
\AtEveryBibitem{
  \clearfield{month}
  \clearfield{edition}
  \clearname{editor}
}
\usepackage[english]{babel}
\usepackage[english = american]{csquotes}
\MakeOuterQuote{"}
\renewbibmacro{in:}{}
\DeclareListFormat{language}{}
\bibliography{qccsccextract.bib}

\usepackage[ruled,vlined,titlenumbered]{algorithm2e}
\usepackage{tabularx}
\makeatletter
\newcommand{\removelatexerror}{\let\@latex@error\@gobble}
\makeatother

\usepackage{pgf}
\usepackage{tikz}
\usetikzlibrary{matrix,chains,positioning,arrows,calc,decorations,plotmarks,patterns, fit,backgrounds}
\usepackage{pgfplots}
\usetikzlibrary{external}                       
\tikzsetexternalprefix{tikz/}                   
\tikzset{external/system call={pdflatex \tikzexternalcheckshellescape -interaction=batchmode -jobname "\image" "\texsource"}}
\pgfplotsset{compat=newest}
\tikzstyle{help lines}=[black!20,dashed]

\pgfplotscreateplotcyclelist{mylist}{red,blue,black,yellow,brown}
\input{def.tex}
\renewcommand{\QCCal}{\ensuremath{\ell}}

\begin{document}
\title{Construction of Quasi-Cyclic Product Codes}

\IEEEoverridecommandlockouts

\author{\IEEEauthorblockN{Alexander Zeh}\thanks{A. Zeh has been supported by the German research council (Deutsche Forschungsgemeinschaft, DFG) under grant Ze1016/1-1. S. Ling has been supported by NTU Research Grant M4080456.}
\IEEEauthorblockA{Computer Science Department\\
Technion---Israel Institute of Technology\\
Haifa, Israel\\
\texttt{alex@codingtheory.eu}
}
\and
\IEEEauthorblockN{San Ling}
\IEEEauthorblockA{Division of Mathematical Sciences, School of Physical \& \\
Mathematical Sciences, Nanyang Technological University\\
Singapore, Republic of Singapore\\
\texttt{lingsan@ntu.edu.sg}}
}

\maketitle

\begin{abstract}
Linear quasi-cyclic product codes over finite fields are investigated. Given the generating set in the form of a reduced Gröbner basis of a quasi-cyclic component code and the generator polynomial of a second cyclic component code, an explicit expression of the basis of the generating set of the quasi-cyclic product code is given. Furthermore, the reduced Gröbner basis of a one-level quasi-cyclic product code is derived.
\end{abstract}


\section{Introduction}
A linear block code of length $\QCCl \QCCm $ over a finite field $\Fq$ is a quasi-cyclic code if every cyclic shift of a codeword by $\QCCl$ positions, for some integer $\QCCl$ between one and $\QCCl \QCCm$, results in another codeword. Quasi-cyclic codes are a natural generalization of cyclic codes (where $\QCCl = 1$), and have a closely linked algebraic structure. In contrast to cyclic codes, quasi-cyclic codes are known to be asymptotically good (see Chen--Peterson--Weldon~\cite{chen_results_1969}). Several such codes have been discovered with the highest minimum distance for a given length and dimension (see Gulliver--Bhargava~\cite{gulliver_best_1991} as well as Chen's and Grassl's databases~\cite{chen_database_2014,Grassl_Codetables}). Several good LDPC codes are quasi-cyclic (see e.g.~\cite{butler_bounds_2013}) and the connection to convolutional codes was investigated among others in~\cite{solomon_connection_1979,esmaeili_link_1998,lally_algebraic_2006}.

Recent papers of Barbier~\textit{et al.} \cite{barbier_quasi-cyclic_2012, barbier_decoding_2013}, Lally--Fitzpatrick~\cite{lally_algebraic_2001, lally_quasicyclic_2003, lally_algebraic_2006}, Ling--Solé~\cite{ling_algebraic_2001, ling_algebraic_2003, ling_algebraic_2005}, Semenov--Trifonov~\cite{semenov_spectral_2012}, Güneri--Özbudak~\cite{guneri_bound_2012} and ours~\cite{zeh_decoding_2014} discuss different aspects of the algebraic structure of quasi-cyclic codes including lower bounds on the minimum Hamming distance and efficient decoding algorithms.

The focus of this paper is on a simple method to combine two given quasi-cyclic codes into a product code. More specifically, we give a description of a quasi-cyclic product code when one component code is quasi-cyclic and the second one is cyclic.

The work of Wasan~\cite{wasan_quasi_1977} first considers quasi-cyclic product codes while investigating the mathematical properties of the wider class of quasi-abelian codes. Some more results were published in a short note by Wasan and Dass~\cite{dass_note_1983}. Koshy proposed a so-called ``circle'' quasi-cyclic product codes in~\cite{koshy_quasi-cyclic_1972}.

Our work considers quasi-cyclic product codes that generalize the results of Burton--Weldon~\cite{burton_cyclic_1965} and Lin--Weldon~\cite{lin_further_1970} (see also~\cite[Chapter 18]{macwilliams_theory_1988}) based on the reduced Gröbner basis representation of Lally--Fitzpatrick~\cite{lally_algebraic_2001} of the quasi-cyclic component code. 
We derive a representation of the generating set of a quasi-cyclic product code, where one component code is quasi-cyclic and the other is cyclic (in Thm.~\ref{theo_QCCTimesCYC}) and we give a reduced Gröbner basis for the special class of one-level quasi-cyclic product codes (in Thm.~\ref{theo_OneLevelQC}).

The paper is structured as follows. In Section~\ref{sec_Preliminaries}, we give necessary preliminaries on quasi-cyclic codes over finite fields. We outline relevant basics of the reduced Gröbner basis representation of Lally--Fitzpatrick~\cite{lally_algebraic_2001}. Furthermore, the special class of $r$-level quasi-cyclic codes is defined in this section. Section~\ref{sec_ProductQCCQCC} contains the main result on quasi-cyclic product codes, where the row-code is quasi-cyclic and the column-code is cyclic. Moreover, an explicit expression of the basis of a 1-level quasi-cyclic product code is derived in Section~\ref{sec_ProductQCCQCC}. For illustration, we explicitly give an example of a binary $2$-quasi-cyclic product code in Section~\ref{sec_Example}.
Section~\ref{sec_Conclusion} concludes this paper.

\section{Preliminaries} \label{sec_Preliminaries}
Let $\Fq$ denote the finite field of order $q$ and $\Fqx$ the polynomial ring over $\Fq$ with indeterminate $X$. Let $a, b$ with $b > a$ be two positive integers and denote by $\interval{a,b}$ the set of integers $\{a,a+1,\dots,b-1\}$ and by $\interval{b}=\interval{0,b}$. A vector of length $n$ is denoted by a lowercase bold letter as $\vec{v} = (v_0 \ v_1 \ \cdots \ v_{n-1})$ and an $m \times n$ matrix is denoted by a capital bold letter as $\M{M}=(m_{i,j})_{i \in \interval{m}}^{j \in \interval{n}}$. 

A linear \LINQCC{\QCCl}{\QCCm}{\QCCk}{\QCCd}{q} code $\QCC$ of length $\QCCl \QCCm$, dimension $\QCCk$ and minimum Hamming distance $\QCCd$ over $\Fq$ is $\QCCcyc$-quasi-cyclic if every cyclic shift by $\QCCcyc$ of a codeword is again a codeword of $\QCC$, more explicitly if:\\[.4ex]
\begingroup
\arraycolsep=.4pt
\begin{tabular}[htb]{llll}
$ (c_{0,0} \cdots c_{\QCCcyc-1,0}$ &  $c_{0,1} \cdots c_{\QCCcyc-1,1}$ & ... & $c_{\QCCcyc-1,\QCClen-1}) \in \QCC $ \\
& $ \Rightarrow $\\
$(c_{0,\QCClen-1} \cdots c_{\QCCcyc-1,\QCClen-1}$ & $c_{0,0} \cdots c_{\QCCcyc-1,0}$ & ... & $c_{\QCCcyc-1,\QCClen-2}) \in \QCC$.
\end{tabular}\\[1ex]
\endgroup
We can represent a codeword of an \LINQCC{\QCCl}{\QCCm}{\QCCk}{\QCCd}{q} $\QCCcyc$-quasi-cyclic code as $\mathbf{c}(X) = (c_0(X) \ c_1(X) \ \cdots \ c_{\QCCcyc-1}(X)) \in \Fqx^{\ell} $, where
\begin{equation} \label{eq_UnivariatePolyCodeword}
c_i(X) \defeq \sum_{j=0}^{\QCClen-1} c_{i,j} X^{j}, \quad \forall i \in \interval{\QCCcyc}.
\end{equation}
Then, the defining property of $\QCC$ is that each component $c_i(X)$ of $\mathbf{c}(X)$ is closed under multiplication by $X$ and reduction modulo $X^{\QCClen}-1$.
\begin{lemma} \label{lem_VectorToUnivariatePoly}
Let $(c_0(X) \ c_1(X) \ \cdots \ c_{\QCCcyc-1}(X))$ be a codeword of an $\QCCcyc$-quasi-cyclic code $\QCC$ of length $\QCClen \QCCcyc$, where the components are defined as in~\eqref{eq_UnivariatePolyCodeword}. Then a codeword in $\QCC$ represented as one univariate polynomial of degree smaller than $\QCClen \QCCcyc$ is
\begin{equation} \label{eq_VectorToUnivariatePoly}
c(X) = \sum_{i=0}^{\QCCcyc-1} c_i(X^{\QCCcyc})X^i.
\end{equation}
\end{lemma}
\begin{proof}
Substitute~\eqref{eq_UnivariatePolyCodeword} into~\eqref{eq_VectorToUnivariatePoly}:
\begin{align*}
c(X) & = \sum_{i=0}^{\QCCcyc-1} c_i(X^{\QCCcyc})X^i = \sum_{i=0}^{\QCCcyc-1} \sum_{j=0}^{\QCClen-1} c_{i,j} X^{j\QCCcyc+i}.
\end{align*}
\end{proof}
\vspace{-.3cm}
Lally and Fitzpatrick~\cite{lally_construction_1999, lally_algebraic_2001} showed that this enables us to see a quasi-cyclic code as an $R$-submodule of the algebra $R^{\QCCcyc}$, where $R = \Fqx/\langle X^{\QCClen}-1 \rangle$. The code $\QCC$ is the image of an $\Fqx$-submodule $\tilde{\QCC}$ of $\Fqx^{\QCCl}$ containing $\basis = \langle (X^{\QCCm}-1)\mathbf{e}_j, j \in \interval{\QCCl} \rangle$ (where $\mathbf{e}_j$ is the standard basis vector with one in position $j$ and zero elsewhere) under the natural homomorphism
\begin{align*}
\phi: \; \Fqx^{\QCCl} & \rightarrow  R^{\QCCcyc} \\
 (c_0(X) \  \cdots \  c_{\QCCcyc - 1}(X)) & \mapsto (c_0(X) + \langle X^{\QCClen} \minus 1 \rangle \ \cdots \\
& \qquad \qquad  c_{\QCCcyc \minus 1}(X) +\langle X^{\QCClen} \minus 1 \rangle ).
\end{align*}
It has a generating set of the form $\{ \mathbf{a}_i, i \in \interval{z}, (X^{\QCCm}-1)\mathbf{e}_j, j \in \interval{\QCCl} \}$, where $\mathbf{a}_i \in \Fqx^{\QCCl}$ and $z \leq \QCCl$ (see e.g.~\cite[Chapter 5]{cox_using_1998} for further information). Therefore, its generating set can be represented as a matrix with entries in $\Fqx$:
\begin{equation} \label{eq_GeneratorWithBasis}
\arraycolsep=1pt
\mathbf{M}(X) = 
\begin{pmatrix}
a_{0,0}(X) & a_{0,1}(X) & \cdots & a_{0,\QCCl-1}(X) \\
a_{1,0}(X) & a_{1,1}(X) & \cdots & a_{1,\QCCl-1}(X) \\
 \vdots & \vdots & \ddots & \vdots \\
a_{z-1,0}(X) & a_{z-1,1}(X) & \cdots & a_{z-1,\QCCl-1}(X) \\
X^{\QCCm}-1 &  &  \\
& X^{\QCCm}-1 & \multicolumn{2}{c}{\bigzero} \\
\multicolumn{2}{c}{\bigzero} & \ddots \\
& & & X^{\QCCm}-1
\end{pmatrix}.
\end{equation}
Every matrix $\mathbf{M}(X)$ as in~\eqref{eq_GeneratorWithBasis} of the preimage $\tilde{\QCC}$ can be transformed into a reduced Gröbner basis (RGB) with respect to the position-over-term order (POT) in $\Fqx^{\QCCcyc}$ (see~\cite{lally_construction_1999, lally_algebraic_2001}).
This basis can be represented in the form of an upper-triangular $\ell \times \ell$ matrix with entries in $\Fqx$ as follows: 
\begin{equation} \label{def_GroebBasisMatrix}
\mathbf{G}(X) =
\begin{pmatrix}
g_{0,0}(X) & g_{0,1}(X) & \cdots  & g_{0,\QCCcyc-1}(X) \\
 & g_{1,1}(X) & \cdots & g_{1,\QCCcyc-1}(X) \\
\multicolumn{2}{c}{\bigzero}& \ddots & \vdots \\
 &  & & g_{\QCCcyc-1,\QCCcyc-1}(X)
\end{pmatrix},
\end{equation}
where the following conditions must be fulfilled:\\[1ex]
\begin{tabular}[htb]{lrll}
1) & $g_{i,j}(X)$ & $= 0,$ & $\forall 0 \leq j < i < \QCCcyc$,\\
2) & $\deg g_{j,i}(X)$ & $ < \deg g_{i,i}(X),$ & $ \forall j < i, i \in \interval{\QCCcyc}$,\\
3) & $g_{i,i}(X)$ & $| \hspace{.2cm}  (X^{\QCClen}-1),$ & $\forall i \in \interval{\QCCcyc}$,\\
4) & if $g_{i,i}(X)$ & $=X^{\QCClen}-1$ then \\ 
& $g_{i,j}(X)$ & $=0,$ & $ \forall j \in \interval{i+1,\QCCcyc}$.
\end{tabular}\\[1ex]
The rows of $\mathbf{G}(X)$  with $g_{i,i}(X) \neq X^{\QCClen}-1$ (i.e., the rows that do not map to zero under $\phi$) are called the reduced generating set of the quasi-cyclic code $\QCC$.
A codeword of $\QCC$ can be represented as $\mathbf{c}(X) = \mathbf{i}(X) \mathbf{G}(X)$ and it follows that $\QCCk = \QCClen \QCCcyc - \sum_{i=0}^{\QCCcyc-1} \deg g_{i,i}(X)$.
Let us recall the following definition (see also~\cite[Thm. 3.2]{lally_construction_1999}).
\begin{definition}[$\level$-level Quasi-Cyclic Code] \label{def_LevelQC}
We call an $\QCCl$-quasi-cyclic code $\QCC$ of length $\QCCl \QCCm$ an $\level$-level quasi-cyclic code if there is an index $\level \in \interval{\QCCl}$ for which the RGB/POT matrix as defined in~\eqref{def_GroebBasisMatrix} is such that $g_{\level-1,\level-1}(X) \neq X^{\QCCm}-1$ and $g_{\level,\level}(X) = \dots = g_{\QCCl-1,\QCCl-1}(X) = X^{\QCCm}-1$.
\end{definition}
We recall \cite[Corollary 3.3]{lally_construction_1999} for the case of a $1$-level quasi-cyclic code in the following.
\begin{corollary}[$1$-level Quasi-Cyclic Code] \label{cor_OneLevelQC}
The generator matrix in RGB/POT form of a $1$-level $\QCCl$-quasi-cyclic code $\QCC$ of length $\QCCl \QCCm$ is:
\begin{equation*}
\mathbf{G}(X) =
\begin{pmatrix}
g(X) & g(X) f_{1}(X)  & \cdots  & g(X)f_{\QCCl-1}(X)
\end{pmatrix},
\end{equation*}
where $g(X) | (X^{\QCCm}-1)$ and $f_{1}(X), \dots, f_{\QCCl-1}(X) \in \Fqx$.
\end{corollary}
To describe quasi-cyclic codes explicitly, we need to recall the following facts of \textit{cyclic} codes. A $q$-cyclotomic coset $\coset{i}{\QCClen}$ is defined as:
$ \coset{i}{\QCClen} \defeq \big\{ iq^j \mod \QCClen \, \vert \, j \in \interval{a} \big\}$,
where $a$ is the smallest positive integer such that $iq^{a} \equiv i \bmod \QCClen$.
The minimal polynomial in $\Fqx$ of the element $\alpha^i \in \F{q^{\QCCext}}$ is given by
\begin{equation} \label{eq_MinPoly}
\minpoly{i}{\QCClen} = \prod_{j \in \coset{i}{\QCClen} } (X-\alpha^j).
\end{equation}
The following fact is used in Section~\ref{sec_ProductQCCQCC}.
\begin{fact} \label{fact_ModuloBlaBla}
Let four nonzero integers $y, a, \ell, m$ be such that 
\begin{equation*} 
y \equiv a \ell \mod m \ell
\end{equation*}
holds. Then $\ell \mid y$ and $y/\ell \equiv a \mod m$.
\end{fact}

\section{Quasi-Cyclic Product Code} \label{sec_ProductQCCQCC}
Throughout this section we consider a linear product code $\QCCa \otimes \QCCb$, where $\QCCa$ is
the row-code and $\QCCb$ the column-code, respectively.
Furthermore,  w.l.o.g. let $\QCCa$ be an \LINQCC{\QCCal}{\QCCam}{\QCCak}{\QCCad}{q} $\QCCal$-quasi-cyclic code
with reduced Gröbner basis in POT form as defined in~\eqref{def_GroebBasisMatrix}:
\begin{equation} \label{eq_GroebMatrixCodeA}
\genmat[A] = 
\begin{pmatrix} 
\gen[A][0][0] & \gen[A][0][1] & \cdots & \gen[A][0][\QCCal-1]\\
 & \gen[A][1][1] & \cdots & \gen[A][1][\QCCal-1] \\
\multicolumn{2}{c}{\bigzero}& \ddots & \vdots \\
 &  & & \gen[A][\QCCal-1][\QCCal-1] 
\end{pmatrix},
\end{equation}
and let $\QCCb$ be an \LIN{\QCCbm}{\QCCbk}{\QCCbd}{q} cyclic code with
generator polynomial $\gen[B]$ of degree $\QCCbm - \QCCbk $.

Throughout the paper, we assume that $\gcd(\QCCal \QCCam, \QCCbm) = 1$ and we furthermore assume that the two integers $\inta$ and $\intb$ are such that 
\begin{equation} \label{eq_BEzoutRel}
\inta \QCCal \QCCam + \intb \QCCbm = 1.
\end{equation}
We recall the lemma of Wasan~\cite{wasan_quasi_1977}, that generalizes the result of Burton--Weldon~\cite[Theorem I]{burton_cyclic_1965} for cyclic product codes to the case of an $\QCCal$-quasi-cyclic product code of an $\QCCal$-quasi-cyclic code $\QCCa$ and a cyclic code $\QCCb$.
A codeword of $\QCCa \otimes \QCCb$ represented as univariate polynomial $c(X)$ can then be obtained from the matrix representation  $(m_{i,j})_{i \in \interval{\QCCbm}}^{j \in \interval{\QCCal \QCCam}}$ as follows:
\begin{equation} \label{eq_OneUnivariatePolyProduct}
c(X) \equiv \sum_{i=0}^{\QCCbm-1} \sum_{j=0}^{\QCCal \QCCam-1}  m_{i,j} X^{\map{i}{j}}  \mod X^{\QCCal \QCCam \QCCbm}-1, 
\end{equation}
where 
\begin{equation} \label{def_MappingMatrixPolyQCCQCC}
\map{i}{j} \defeq i \inta \QCCal  \QCCam \QCCal + j \intb \QCCbm  \mod  \QCCal \QCCam \QCCbm.
\end{equation}
\begin{lemma}[Mapping to a Univariate Polynomial~\cite{wasan_quasi_1977}] \label{lem_MappingToUnivariatePolyQCC}
Let $\QCCa$ be an $\QCCal$-quasi-cyclic code of length $\QCCal \QCCam$ and let $\QCCb$ be a cyclic code of length $\QCCbm$. The product code $\QCCa \otimes \QCCb$ is an $\QCCal$-quasi-cyclic code of length $\QCCal \QCCam \QCCbm$ if $\gcd(\QCCal\QCCam, \QCCbm) = 1$.
\end{lemma}
\begin{proof}
Let $(m_{i,j})_{i \in \interval{\QCCbm}}^{j \in \interval{\QCCal \QCCam}}$ be a codeword of the product code $\QCCa \otimes \QCCb$, where each row is a codeword of $\QCCa$ and each column is a codeword of $\QCCb$. The entry $m_{i,j}$ is the coefficient $c_{\map{i}{j}}$ of the codeword $\sum_{i} c_i X^i$ as in~\eqref{eq_OneUnivariatePolyProduct}.
In order to prove that $\QCCa \otimes \QCCb$ is $\QCCal$-quasi-cyclic it is sufficient to show that a shift by $\QCCal$ positions of a codeword serialized to a univariate polynomial by~\eqref{def_MappingMatrixPolyQCCQCC} of $\QCCa \otimes \QCCb$ is again a codeword of $\QCCa \otimes \QCCb$. 

A shift by $\QCCal$ in each row and a shift by one each column clearly gives a codeword in $\QCCa \otimes \QCCb$, because $\QCCa$ is $\QCCal$-quasi-cyclic and $\QCCb$ is cyclic. 
With
\begin{align*}
& \map{i+1}{j+\QCCal} \\
& \equiv (i+1) \inta \QCCal \QCCam \QCCal + (j+\QCCal) \intb \QCCbm \mod \QCCal \QCCam \QCCbm \\
& \equiv i \inta \QCCal \QCCam \QCCal + j \intb \QCCbm + \QCCal  (\inta \QCCal \QCCam  + \intb \QCCbm) \mod \QCCal \QCCam \QCCbm \\
& \equiv \map{i}{j} + \QCCal  \mod \QCCal \QCCam \QCCbm, 
\end{align*}
we obtain an $\QCCal$-quasi-cyclic shift of the univariate codeword obtained by~\eqref{eq_OneUnivariatePolyProduct} and~\eqref{def_MappingMatrixPolyQCCQCC}.
\end{proof}
Instead of representing a codeword of $\QCCa \otimes \QCCb$ as one univariate polynomial as in~\eqref{eq_OneUnivariatePolyProduct}, we want to represent it as $\QCCal$ univariate polynomials as defined in~\eqref{eq_UnivariatePolyCodeword}. 
\begin{lemma}[Mapping to $\QCCal$ Univariate Polynomials] \label{lem_MappingUnivariateQCCQCC}
Let $\QCCa$ be an $\QCCal$-quasi-cyclic code of length $\QCCal \QCCam$ and let $\QCCb$ be a cyclic code of length $\QCCbm$.
Let the matrix $(m_{i,j})_{i \in \interval{\QCCbm}}^{j \in \interval{\QCCal \QCCam}}$ be a codeword of $\QCCa \otimes \QCCb$, where each row is in $\QCCa$ and each column is in $\QCCb$. The $\QCCal $ univariate polynomials of the corresponding codeword $(c_{0}(X) \ c_{1}(X) \ \cdots \ c_{\QCCal-1}(X)) $, where each component is defined as in~\eqref{eq_UnivariatePolyCodeword}, are given by:
\begin{equation} \label{eq_ComponentExpression}
\begin{split}
c_{h}(X) & \equiv X^{h(-\inta \QCCam)} \cdot \sum_{i=0}^{\QCCbm-1} \sum_{j=0}^{\QCCam-1} m_{i,j\QCCal+h} X^{ \mapb{i}{j}} \\
& \qquad \qquad \mod X^{\QCCam \QCCbm} - 1, \quad \forall h \in \interval{\QCCal},
\end{split}
\end{equation}
where
\begin{equation} \label{eq_MappingSubCodewordQCCQCC}
\mapb{i}{j} \equiv i \inta \QCCal \QCCam + j \intb \QCCbm \mod \QCCam \QCCbm.
\end{equation}
\end{lemma}
\begin{proof}
From Fact~\ref{fact_ModuloBlaBla} we have for the exponents in~\eqref{eq_ComponentExpression}:
\begin{align} 
& \mapb{i}{j} +  h ( -\inta \QCCam)  \equiv i \inta \QCCal \QCCam + j \intb \QCCbm \mod \QCCam \QCCbm  \nonumber \\
& \Leftrightarrow \nonumber \\
& \QCCal \big( \mapb{i}{j} +  h ( -\inta \QCCam ) \big)  \nonumber \\
& \equiv \QCCal(i \inta \QCCal \QCCam + j \intb \QCCbm +  h ( -\inta \QCCam ) )  \mod \QCCal \QCCam \QCCbm.  \label{eq_FinalBigToSmallMatrix}
\end{align}
With $ -\inta \QCCal \QCCam = \intb \QCCbm-1  $, we can rewrite~\eqref{eq_FinalBigToSmallMatrix}:
\begin{align}
\QCCal \big( \mapb{i}{j} +  h ( -\inta \QCCam ) \big) & = \QCCal \mapb{i}{j} + \QCCal h ( -\inta \QCCam ) \nonumber  \\ 
& = \QCCal \mapb{i}{j} + h \intb \QCCbm  -h , \nonumber 
\end{align}
and this gives with $\mapb{i}{j}$ as in~\eqref{eq_MappingSubCodewordQCCQCC} and $\map{i}{j}$ as in~\eqref{def_MappingMatrixPolyQCCQCC}:
\begin{align}
& \QCCal \mapb{i}{j} +  h \intb \QCCbm  - h \nonumber  \\
& \equiv \QCCal (i \inta \QCCal \QCCam  + j \intb \QCCbm) + h \intb \QCCbm - h    \nonumber \\
& \equiv  \QCCal i \inta \QCCal \QCCam + (j \QCCal+h) \intb \QCCbm - h \mod \QCCal \QCCam \QCCbm \nonumber \\
& = \map{i}{j\QCCal+h} - h. \label{eq_ExpressionForComponennt}
\end{align}
Inserting~\eqref{eq_ExpressionForComponennt} in~\eqref{eq_VectorToUnivariatePoly} of  Lemma~\ref{lem_VectorToUnivariatePoly} leads to:
\begin{align}
c(X) & = \sum_{h=0}^{\QCCal-1} c_h(X^{\QCCal}) X^h \nonumber \\
 & = \sum_{h=0}^{\QCCal-1} \sum_{i=0}^{\QCCbm-1} \sum_{j=0}^{\QCCam-1} m_{i,j\QCCal+h} X^{\map{i}{j\QCCal+h}} \nonumber \\
& = \sum_{i=0}^{\QCCbm-1} \sum_{j=0}^{\QCCal \QCCam-1} m_{i,j} X^{\map{i}{j}},
\end{align}
which equals~\eqref{eq_OneUnivariatePolyProduct}.
\end{proof}
The mapping $\mapb{i}{j}$ from~\eqref{eq_MappingSubCodewordQCCQCC} of the $\QCCal$ submatrices $(m_{i,j\QCCal})_{i \in \interval{\QCCbm}}^{j \in \interval{\QCCam}}, (m_{i,j\QCCal+1})_{i \in \interval{\QCCbm}}^{j \in \interval{\QCCam}}, \dots, (m_{i,j\QCCal+\QCCal-1})_{i \in \interval{\QCCbm}}^{j \in \interval{\QCCam}}$ to the $\QCCal$ univariate polynomials $c_0(X), c_1(X), \dots, c_{\QCCal-1}(X) $ is the same as the one used to map the codeword of a cyclic product code from its matrix representation to a polynomial representation (see \cite[Thm. 1]{burton_cyclic_1965}).

In Fig.~\ref{fig_CyclicQCC}, we illustrate the $\map{i}{j}$ as in~\eqref{def_MappingMatrixPolyQCCQCC} for $\inta=1$, $\QCCal = 2$, $\QCCam = 17$ and $\intb=-11$, $\QCCbm=3$. Subfigure~\ref{fig_QCCCYCfull} shows the values of $\map{i}{j}$. The two submatrices $(m_{i,j2})$ and $(m_{i,j2+1})$ for $i \in \interval{3}$ and $j \in \interval{17}$ are shown in Subfigure~\ref{fig_QCCCYCb}. Subfigure~\ref{fig_QCCCYCc} contains the coefficients of the two univariate polynomials $c_0(X)$ and $c_1(X)$, where $(c_0(X) \ c_1(X))$ is a codeword of the $2$-quasi-cyclic product code of length $102$.
\newcommand{\mysize}{7mm}
\begin{figure*}[htb]
\centering
\subfigure[The $3 \times (2 \cdot 17)$ codeword  matrix $(m_{i,j})$ of the $2$-quasi-cyclic product code $\QCCa \otimes \QCCb$. Each entry contains the index of the coefficient $c_i$ of the univariate polynomial $c(X) = \sum_{i=0}^{101} c_i X^i \in \QCCa \otimes \QCCb $.]{\resizebox{1.8\columnwidth}{!}{\includegraphics{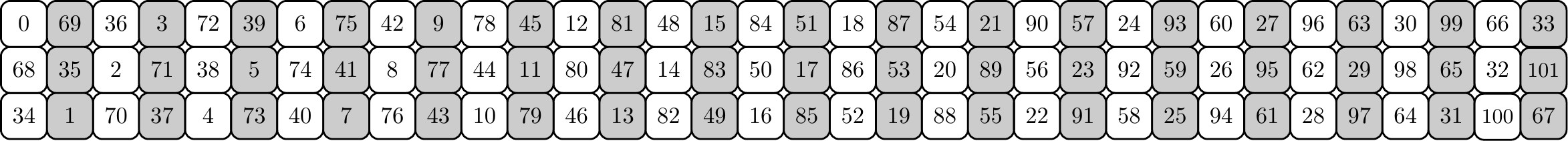}}\label{fig_QCCCYCfull}}\\
\subfigure[The two submatrices $(m_{i,2j})_{i \in \interval{3}}^{i \in \interval{17}}$ and $(m_{i,2j+1})_{i \in \interval{3}}^{i \in \interval{17}}$ with entries that are the coefficients of $c_0(X^2)$ and $X c_1(X^2)$.]{\tikzsetnextfilename{Product17-3-b}\resizebox{1.8\columnwidth}{!}{\includegraphics{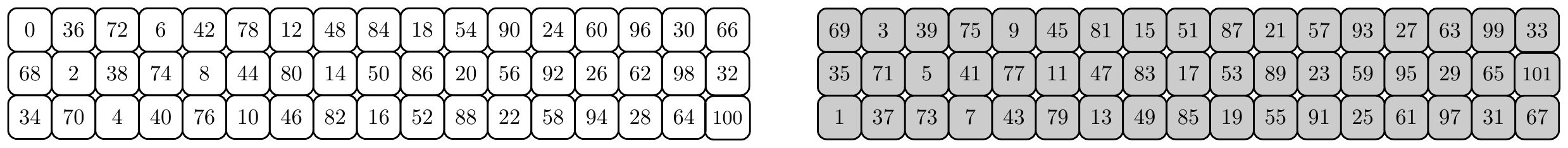}}\label{fig_QCCCYCb}}\\
\subfigure[The left submatrix contains the coefficients $c_{0,i}$ of the univariate polynomials $c_0(X)$ (the right one contains $c_{1,i}$ of $c_1(X)$, respectively).]{\tikzsetnextfilename{Product17-3-c}\resizebox{1.8\columnwidth}{!}{\includegraphics{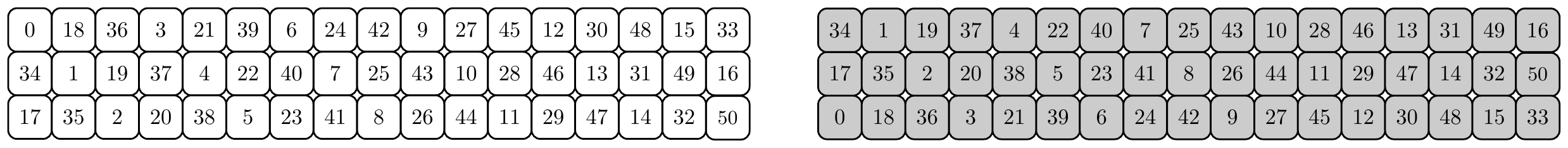}}\label{fig_QCCCYCc}}
\label{fig_CyclicQCC}
\caption{Illustration of the mapping $\map{i}{j}$ (as defined in~\eqref{def_MappingMatrixPolyQCCQCC}) from a codeword of a quasi-cyclic product code represented as matrix to a polynomial representation. The product code $\QCCa \otimes \QCCb$ is $2$-quasi-cyclic. The row-code $\QCCa$ is $2$-quasi-cyclic and has length $\QCCal\QCCam = 2 \cdot 17$ and the column-code $\QCCb$ is cyclic and has length $\QCCbm = 3$ (Subfigure~\ref{fig_QCCCYCfull}, here $\inta = 1$ and $\intb = -11$). The mapping $\mapb{i}{j}$ (as in~\eqref{eq_MappingSubCodewordQCCQCC}) to two univariate polynomials is illustrated in Subfigure~\ref{fig_QCCCYCb} and Subfigure~\ref{fig_QCCCYCc}.}
\end{figure*}

The following theorem gives the basis representation of a quasi-cyclic product code, where the row-code is quasi-cyclic and the column-code is cyclic.
\begin{theorem}[Quasi-Cyclic Product Code] \label{theo_QCCTimesCYC}
Let $\QCCa$ be an $\LINQCC{\QCCal}{\QCCam}{\QCCak}{\QCCad}{q}$ $\QCCal$-quasi-cyclic code with generator matrix $\genmat[A] \in \Fqx^{\QCCal \times \QCCal}$  as in~\eqref{eq_GroebMatrixCodeA} and let $\QCCb$ be an $\LIN{\QCCbm}{\QCCbk}{\QCCbd}{q}$ cyclic code with generator polynomial $\gen[B] \in \Fqx$.

Then the $\QCCal$-quasi-cyclic product code $\QCCa \otimes \QCCb$ has a generating matrix of the following (unreduced) form:
\begin{equation} \label{eq_UnReducedBasis}
\genmat = 
\begin{pmatrix}
\genmat[0] \\
\genmat[1] \\
\end{pmatrix},
\end{equation}
where 
\begin{equation} \label{eq_UnReducedBasisPart1}
\begin{split}
& \genmat[0] = \genarg{B}{X^{\inta \QCCal \QCCam }} \cdot \\
& \begin{pmatrix}
\genarg{A}[0][0]{X^{\intb \QCCbm}}  & \genarg{A}[0][1]{X^{\intb \QCCbm}}  & \cdots & \genarg{A}[0][\QCCal-1]{X^{\intb \QCCbm}} \\
 & \genarg{A}[1][1]{X^{\intb \QCCbm}} & \cdots & \genarg{A}[1][\QCCal-1]{X^{\intb \QCCbm}} \\
\multicolumn{2}{c}{\bigzerob} & \ddots & \vdots \\
& &  & \genarg{A}[\QCCal-1][\QCCal-1]{X^{\intb \QCCbm}} \\
\end{pmatrix},\\
& \cdot \diag \begin{pmatrix}
1, X^{-\inta \QCCam}, X^{-2 \inta \QCCam}, \dots, X^{- (\QCCal-1) \inta \QCCam} 
\end{pmatrix} 
\end{split}
\end{equation}
and  
\begin{equation} \label{eq_UnReducedBasisPart2}
\genmat[1] = (X^{\QCCam \QCCbm}-1) \mathbf{I}_{\QCCal}, 
\end{equation}
where $\mathbf{I}_{\QCCal}$ is the $\QCCal \times \QCCal$ identity matrix.
\end{theorem}
\begin{proof}
We first give an explicit expression for each component of a codeword $(c_{0}(X) \ c_{1}(X) \ \cdots \ c_{\QCCal-1}(X))$ in $\QCCa \otimes \QCCb $ depending on the components of a codeword  $(a_{0}(X) \ a_{1}(X) \ \cdots \ a_{\QCCal-1}(X) )$ of the row-code $\QCCa$ and depending the column-code $\QCCb$ based on the expression of Lemma~\ref{lem_MappingUnivariateQCCQCC}.
Let the $\QCCbm \times \QCCal \QCCam $ matrix $(m_{i,j})$ be a codeword of the $\QCCal$-quasi-cyclic product code $\QCCa \otimes \QCCb$ and let the polynomial 
\begin{equation} \label{eq_ColumnComponentQCCQCC}
a_{i,h}(X) \defeq \sum_{j=0}^{\QCCam-1} m_{i,j\QCCal+h} X^j, \quad \forall h \in \interval{\QCCal}, i \in \interval{\QCCbm}
\end{equation}
denote the $h$th component of a codeword $(a_{i,0}(X) \ a_{i,1}(X) \ \cdots \ a_{i, \QCCal-1}(X))$ in $\QCCa$ in the $i$th row of the matrix $(m_{i,j})$. Denote a codeword $b_{j}(X)$ of $\QCCb$ in the $j$th column by
\begin{equation} \label{eq_RowComponentQCCQCC}
b_{j}(X) = \sum_{i=0}^{\QCCbm-1} m_{i,j} X^i, \quad \forall j \in \interval{\QCCal \QCCam},
\end{equation}
respectively. From~\eqref{eq_ComponentExpression}, we have for the $h$th component of a codeword of the product code $ \QCCa \otimes \QCCb$:
\begin{equation} \label{eq_ProductComponentQCC}
\begin{split}
c_{h}(X) & \equiv X^{h(-\inta \QCCam)} \sum_{i=0}^{\QCCbm-1} \sum_{j=0}^{\QCCam-1} m_{i,j\QCCal+h} X^{ \mapb{i}{j} } \\ 
& \qquad \qquad    \mod X^{\QCCam \QCCbm} - 1, \quad \forall h \in \interval{\QCCal},
\end{split}
\end{equation}
and with $\mapb{i}{j}$ as in~\eqref{eq_MappingSubCodewordQCCQCC} of Lemma~\ref{lem_MappingUnivariateQCCQCC} we can write~\eqref{eq_ProductComponentQCC} explicitly:
\begin{align} 
&  c_{h}(X) \equiv X^{h(-\inta \QCCam)} \sum_{i=0}^{\QCCbm-1} \sum_{j=0}^{\QCCam-1} m_{i,j\QCCal+h}  X^{i \inta \QCCal \QCCam + j\intb \QCCbm} \nonumber \\
& \qquad \qquad \qquad \mod X^{\QCCam \QCCbm} - 1, \quad \forall h \in \interval{\QCCal}. \label{eq_ExplicitComponentProductQCC} 
\end{align}
We define a shifted component: 
\begin{equation} \label{eq_ExplicitComponentProductQCCShifted}
\tilde{c}_{h}(X) \equiv c_{h}(X) X^{h(\inta \QCCam)} \mod X^{\QCCam \QCCbm} \minus 1, \: \forall h \in \interval{\QCCal}. 
\end{equation}
Since 
\begin{align*}
\sum_{i=0}^{\QCCbm-1} \sum_{j=0}^{\QCCam-1} & m_{i,j\QCCal+h} X^{i \inta \QCCal \QCCam + j \intb \QCCbm} \\
& = \sum_{i=0}^{\QCCbm-1} X^{i \inta \QCCal \QCCam}  \sum_{j=0}^{\QCCam-1} m_{i,j\QCCal+h} X^{j \intb \QCCbm} \\
& = \sum_{i=0}^{\QCCbm-1} X^{i \inta \QCCal \QCCam}  a_{i,h}(X^{\intb \QCCbm}), \quad \forall h \in \interval{\QCCal},
\end{align*}
and from~\eqref{eq_ExplicitComponentProductQCCShifted} and in terms of the components of the row-code as defined in~\eqref{eq_ColumnComponentQCCQCC}, we obtain: 
\begin{equation} \label{eq_ExplicitComponentWithRemainderQCCQCC}
\begin{split}
& \tilde{c}_{h}(X) = q_{h}(X) (X^{\QCCam \QCCbm}-1) + \\
& \qquad \qquad \sum_{i=0}^{\QCCbm-1} X^{i \inta \QCCal \QCCam}  a_{i,h}(X^{\intb \QCCbm}), \quad  \forall h \in \interval{\QCCal},
\end{split}
\end{equation}
for some $q_{h}(X) \in \Fqx$. Therefore $\tilde{c}_{h}(X)$ is a multiple of $\sum_{i=0}^h \epsilon_{i}(X) \genarg{A}[i][h]{X^{\intb \QCCbm}}$ for some $\epsilon_i(X) \in \Fqx$.
A codeword $b_{j}(X)$ in $\QCCb$ in the $j$th column of $(m_{i,j})$ is a multiple of $\genarg{B}{X}$ and we obtain:
\begin{align*}
\sum_{i=0}^{\QCCbm-1} \sum_{j=0}^{\QCCal \QCCam-1} & m_{i,j} X^{i \inta \QCCal \QCCam+j\intb\QCCbm} \\
& = \sum_{j=0}^{\QCCal \QCCam-1}  X^{j \intb \QCCbm} \sum_{i=0}^{\QCCbm-1} m_{i,j} X^{i \inta \QCCal \QCCam}   \\
& = \sum_{j=0}^{\QCCal \QCCam-1} X^{j \intb \QCCbm}  b_{j}(X^{\inta \QCCal \QCCam}),
\end{align*}
and therefore $\tilde{c}_h(X)$ is a multiple of $\genarg{B}{X^{\inta \QCCal \QCCam}}$ modulo $X^{\QCCam \QCCbm} - 1$.

Similar to the proof of~\cite[Thm. III]{burton_cyclic_1965}, it can be shown that every shifted component $\tilde{c}_h(X)$ is a multiple of the product of $\genarg{B}{X^{\inta \QCCal \QCCam}}$ and $\sum_{i=0}^h \epsilon_{i} \genarg{A}[i][h]{X^{\intb \QCCbm}} $ modulo $(X^{\QCCam \QCCbm}-1)$. Therefore, we can represent each codeword in $\QCCa \otimes \QCCb$ as:
\begin{equation*}
\begin{split}
& \big(c_0(X) \ c_1(X) \ \cdots \ c_{\QCCal-1}(X) \big) \\
& \qquad = \big(i_0(X) \ i_1(X) \ \cdots \ i_{\QCCal-1}(X) \big) \genmat, 
\end{split}
\end{equation*}
where $\genmat$ is as in~\eqref{eq_UnReducedBasis}.
\end{proof}
The following theorem gives the reduced Gröbner basis (as defined in~\eqref{def_GroebBasisMatrix}) representation of the quasi-cyclic product code from Thm.~\ref{theo_QCCTimesCYC}, where the row-code is a 1-level quasi-cyclic code.
\begin{theorem}[1-Level Quasi-Cyclic Product Code] \label{theo_OneLevelQC}
Let $\QCCa$ be an $\LINQCC{\QCCal}{\QCCam}{\QCCak}{\QCCad}{q}$ 1-level $\QCCal$-quasi-cyclic code with generator matrix in RGB/POT form:
\begin{align}
& \genmat[A]  \nonumber \\
& = \begin{pmatrix}
\gen[A][0][0] & \hspace{.2cm} \gen[A][0][1] &  \hspace{.55cm} \cdots & \hspace{.3cm} \gen[A][0][\QCCal-1] 
\end{pmatrix} \nonumber \\
& = \begin{pmatrix}
\gen[A] & \gen[A] f_{1}^{A}(X) & \cdots & \gen[A] f_{\QCCal-1}^{A}(X)
\end{pmatrix} \label{eq_GenMatrixQCCOneLevel}
\end{align}
as shown in Corollary~\ref{cor_OneLevelQC}. Let $\QCCb$ be an $\LIN{\QCCbm}{\QCCbk}{\QCCbd}{q}$ cyclic code with generator polynomial $\gen[B] \in \Fqx$.

Then a generator matrix of the $1$-level $\QCCal$-quasi-cyclic product code in RGB/POT form is:
\begin{small}
\begin{equation*}
\begin{split}
\genmat = &  \begin{pmatrix}
\gen & \gen f_{1}^{A}(X^{\intb \QCCbm})  & \cdots & \gen f_{\QCCal-1}^{A}(X^{\intb \QCCbm })
\end{pmatrix}\\
& \cdot \diag \begin{pmatrix}
1, X^{-\inta \QCCam}, X^{-2 \inta \QCCam}, \dots, X^{- (\QCCal-1) \inta \QCCam} 
\end{pmatrix},
\end{split}
\end{equation*}
\end{small}
\hspace{-.1cm}where 
\begin{equation} \label{eq_GCDOneLevel}
\gen = \gcd \big( X^{\QCCam \QCCbm}-1, \genarg{A}{X^{\intb \QCCbm}} \genarg{B}{X^{\inta \QCCal \QCCam}} \big) .
\end{equation}
\end{theorem}
\begin{proof}
Let two polynomials $u_0(X), v_0(X) \in \Fqx$ be such that:
\begin{equation} \label{eq_BezoutProductDiag}
\begin{split}
\gen & = u_0(X) \genarg{A}{X^{\intb \QCCbm}} \genarg{B}{X^{\inta \QCCal \QCCam }}\\
& \qquad \qquad + v_0(X) (X^{\QCCam \QCCbm}-1).
\end{split}
\end{equation}
We show now how to reduce the basis representation to the RGB/POT form. We denote a new Row $i$ by $\rowop{i}'$. For ease of notation, we omit the term $\diag(1, X^{-\inta \QCCam}, X^{-2 \inta \QCCam}, \dots ,X^{- (\QCCal-1) \inta \QCCam} )$ and denote by $Y = X^{\intb \QCCbm}$ and $Z = X^{\inta \QCCal \QCCam}$.

We write the basis of the submodule in unreduced form (as in~\eqref{eq_UnReducedBasis}):
\begingroup
\begin{align}
& \begin{pmatrix}
\genarg{A}{Y} \genarg{B}{Z} & \genarg{A}{Y} f_{1}^{A}(Y) \genarg{B}{Z} & \cdots & \\ 
X^{\QCCam \QCCbm}-1 &   \\
& X^{\QCCam \QCCbm}-1 &  \\
\multicolumn{1}{c}{\bigzero} &  & \ddots & \\
\end{pmatrix} \label{eq_StartMatrix}  \\[2ex]
& \rightarrow \rowop{0}' = u_0(X)\rowop{0} + v_0(X) \rowop{1} + v_0(X) f_{1}^{A}(Y) \rowop{2} \nonumber \\ 
& \qquad \qquad \qquad + \dots + v_0(X) f_{\QCCal-1}^{A}(Y) \rowop{\QCCal}  \nonumber \\[2ex]
& \begin{pmatrix}
\gen & \gen f_{1}^{A}(Y) & \cdots & \\
\genarg{A}{Y} \genarg{B}{Z} & \genarg{A}{Y} f_{1}^{A}(Y) \genarg{B}{Z} & \cdots & \\
X^{\QCCam \QCCbm}-1 & \\
& X^{\QCCam \QCCbm}-1 &  \\
\multicolumn{1}{c}{\bigzero} & & \ddots & \\
\end{pmatrix}, \label{eq_MatrixFirstMerge}
\end{align}
where the $i$th entry in new row 0 was obtained using:
\begin{align}
& u_0(X) \genarg{A}{Y} f_{i}^{A}(Y) \genarg{B}{Z} + v_0(X) f_{i}^{A}(Y) (X^{\QCCam \QCCbm} -1) \nonumber \\
& = f_{i}^{A}(Y) \big( u_0(X) \genarg{A}{Y} \genarg{B}{Z} \nonumber \\
& \qquad \qquad + v_0(X) (X^{\QCCam \QCCbm} - 1) \big), \label{eq_PreGCDForm}
\end{align}
and with~\eqref{eq_BezoutProductDiag} we obtain from~\eqref{eq_PreGCDForm}
\begin{align*}
& f_{i}^{A}(Y) \big( u_0(X) \genarg{A}{Y} \genarg{B}{Z} + v_0(X) (X^{\QCCam \QCCbm}-1) \big) \\
& = f_{i}^{A}(Y) \gen.
\end{align*}
Clearly, $\gen$ divides $\genarg{A}{Y} \genarg{B}{Z}$ and it is easy to check that Row 1 of the matrix in~\eqref{eq_MatrixFirstMerge} can be obtained from Row 0 by multiplying by $\genarg{A}{Y} \genarg{B}{Z}/\gen$. 
Therefore, we can omit the linearly dependent Row $1$ in~\eqref{eq_MatrixFirstMerge} and write the reduced basis as:
\begin{align*}
& \begin{pmatrix}
\gen \hspace*{.3cm}  & \gen f_{1}^{A}(X^{\intb \QCCbm}) & \cdots   & \gen f_{\QCCal-1}^{A}(X^{\intb \QCCbm})
\end{pmatrix}, 
\end{align*}
where we omitted the matrix $\diag (1, X^{-\inta \QCCam}, X^{-2 \inta \QCCam}, \dots, X^{- (\QCCal-1) \inta \QCCam})$ for the first row during the proof, but it will only influence the row-operations by a factor.
\endgroup
\end{proof}
Note that~\eqref{eq_GCDOneLevel} is exactly the generator polynomial of a cyclic product code.
A $1$-level $\QCCal$-quasi-cyclic product has rate greater than $(\QCCal-1)/\QCCal$ and is therefore of high practical relevance. The explicit RGB/POT form of the $1$-level quasi-cyclic product code as in Thm.~\ref{theo_OneLevelQC} allows statements on the minimum distance and to develop decoding algorithms.

\section{Example} \label{sec_Example}
We consider a $2$-quasi-cyclic product code with the same parameters as the one illustrated in Fig.~\ref{fig_CyclicQCC}. In this section we investigate a more explicit example to be able to calculate the basis as given in Thm.~\ref{theo_OneLevelQC}.

Let $\QCCa$ be a binary $2$-quasi-cyclic code of length $\QCCal \QCCam = 2\cdot 17 = 34$ and let $\QCCb$ be a cyclic code of length $\QCCbm=3$. 
We have $X^{17}-1=\minpoly{17}{0} \minpoly{17}{1} \minpoly{17}{3} $, where the minimal polynomials are as defined in~\eqref{eq_MinPoly}. Let the generator matrix of $\QCCa$ in RGB/POT form as in~\eqref{def_GroebBasisMatrix} be $ \genmat[A] = \begin{pmatrix} \gen[A][0][0] & \gen[A][0][1] \end{pmatrix}$ where
\begin{align*}
\gen[A][0][0]  & = \minpoly{17}{1}\\ 
 &  = X^8 + X^7 + X^6 + X^4 + X^2 + X + 1, \\
\gen[A][0][1] & = \minpoly{17}{1} \cdot m_0(X)^3 \cdot (X^3+X^2+1) \\
&  = X^{14} + X^{13} + X^{12} + X^{11} + X^8 + 1,
\end{align*}
and $\QCCa$ is a $\LINQCC{17}{2}{9}{11}{2}$ $2$-quasi-cyclic code. 
Let $\alpha$ be a $17$th root of unity in $\Fxsub{2^{8}} \cong \Fxsub{2}/(X^8 + X^4 + X^3 + X^2 + 1)$. Let $\gen[B] = \minpoly{3}{0} = X+1$ be the generator polynomial of the $\LIN{3}{2}{2}{2}$ cyclic code $\CYCb$ and let $\inta = 1$ and $\intb = -11$ be such that~\eqref{eq_BEzoutRel} holds. We have
\begin{align*}
 X^{51}-1 & = \minpoly{51}{0}\minpoly{51}{1}\minpoly{51}{3}\minpoly{51}{5}\minpoly{51}{9}\\
 & \qquad  \minpoly{51}{11}\minpoly{51}{17}\minpoly{51}{19}.
\end{align*}
According to Thm.~\ref{theo_OneLevelQC}, we calculate 
\begin{align*}
f_{1}^{A}(X^{-11 \cdot 3}) &  \equiv f_{0,1}^{A}(X^{18}) = m_0(X^{18})^3 \cdot (X^{54}+X^{36}+1)\\
 & = (X^{18}+1)^3 \cdot (X^{54}+X^{36}+1) \\
 & = X^{108} + X^{54} + X^{18} + 1 \\
 & \equiv X^{18} + X^6 + X^3 + 1 \mod (X^{51}+1), 
\end{align*}
and we obtain the generator matrix $\genmat = (\gen[0][0] \ \gen[0][1])$ of $\QCCa \otimes \QCCb$, where:
\begin{align*}
\gen[0][0] & = \minpoly{51}{0}\minpoly{51}{1}\minpoly{51}{3}\minpoly{51}{9}\minpoly{51}{19}\\
& = X^{33} + X^{32} + X^{30} + X^{27} + X^{25} + X^{23} + X^{20} \\
& \quad + X^{18} + X^{17} + X^{16} + X^{15} + X^{13} + X^{10} + X^{8} \\
& \quad + X^{6} + X^{3} + X + 1.
\end{align*}
With~Thm.~\ref{theo_OneLevelQC}, we obtain:
\begin{align*}
\gen[0][1] & \equiv a_{1}^{A}(X^{-11 \cdot 3}) \gen[0][0] \\
& \equiv X^{50} + X^{48} + X^{45} + X^{43} + X^{41} + X^{39} + X^{36}\\
& \quad  +  X^{34} + X^{32} + X^{29} + X^{27} + X^{26} + X^{25} + X^{23}\\ 
& \quad + X^{22} + X^{21} + X^{19} + X^{18} + X^{17} + X^{16} + X^{15}\\
& \quad + X^{14} + X^{12} + X^{11} + X^{10} + X^{8} + X^{7} + X^{6}\\
& \quad + X^{4} + X  \mod (X^{51}+1).
\end{align*}
\section{Conclusion and Outlook} \label{sec_Conclusion}
Based on the RGB/POT representation of an $\QCCal$-quasi-cyclic code $\QCCa$ and the generator polynomial of a cyclic code $\QCCb$, a basis representation of the $\QCCal$-quasi-cyclic product code $\QCCa \otimes \QCCb$ was proven. The reduced basis representation of the special case of a $1$-generator quasi-cyclic product code was derived.

The general case of the basis representation of an $\QCCl_A \QCCbl$-quasi cyclic product code from an $\QCCl_A$-quasi-cyclic code $\QCCa$ and an $\QCCbl$-quasi-cyclic code $\QCCb$ as well as the reduction of the basis remains an open future work. Furthermore, a technique to bound the minimum distance of a given quasi-cyclic code by embedding it into a product code similar to~\cite{zeh_decoding_2012} seems to be realizable.

\printbibliography
\end{document}

%% file: def.tex
\newtheorem{definition}{Definition}
\newtheorem{theorem}[definition]{Theorem}
\newtheorem{lemma}[definition]{Lemma}

\newtheorem{corollary}[definition]{Corollary}
\newtheorem{fact}[definition]{Fact}
\DeclareMathOperator{\defi}{def}
\newcommand{\defeq}{\overset{\defi}{=}}

\newcommand{\F}[1]{\mathbb F_{#1}}
\newcommand{\Fq}{\F{q}}
\newcommand{\Fxsub}[1]{\ensuremath{\mathbb{F}_{#1}[X]}}
\newcommand{\Fqx}{\Fxsub{q}}

\renewcommand{\vec}[1]{\mathbf #1}
\newcommand{\M}[2][\empty]{
  \ifthenelse{\equal{#1}{\empty}}
    {\ensuremath{\mathbf{#2}}}
    {\ensuremath{{\mathbf{#2}}_{#1}}}
}
\newcommand{\SET}[1]{\ensuremath{\mathsf{#1}}}

\newcommand{\LIN}[4]{\ensuremath{[#1,#2,#3]_{#4}}}

\newcommand{\defset}[2][\empty]{
  \ifthenelse{\equal{#1}{\empty}}
    {\ensuremath{\SET{D}_{#2}}}
    {\ensuremath{\SET{D}^{[#1]}_{#2}}}
}

\newcommand{\QCCa}{\ensuremath{\mathcal{A}}}
\newcommand{\QCCam}{\ensuremath{m_{A}}} 
\newcommand{\QCCal}{\ensuremath{\ell_{A}}}
\newcommand{\QCCak}{\ensuremath{k_{A}}} 
\newcommand{\QCCad}{\ensuremath{d_{A}}} 
\newcommand{\QCCb}{\ensuremath{\mathcal{B}}}
\newcommand{\QCCbm}{\ensuremath{m_{B}}} 
\newcommand{\QCCbl}{\ensuremath{\ell_{B}}}
\newcommand{\QCCbk}{\ensuremath{k_{B}}} 
\newcommand{\QCCbd}{\ensuremath{d_{B}}} 
\newcommand{\QCC}{\ensuremath{\mathcal{C}}}
\newcommand{\QCCm}{\ensuremath{m}} 
\newcommand{\QCCl}{\ensuremath{\ell}}
\newcommand{\QCCk}{\ensuremath{k}} 
\newcommand{\QCCd}{\ensuremath{d}} 

\newcommand{\CYCb}{\ensuremath{\mathcal{B}}}

\newcommand{\inta}{\ensuremath{a}}
\newcommand{\intb}{\ensuremath{b}}

\renewcommand{\tilde}{\widetilde}
\renewcommand{\bar}{\overline}

\newcommand{\LINQCC}[5]{\ensuremath{[#1\cdot#2,#3,#4]_{#5}}}

\newcommand{\QCClen}{\ensuremath{m}}
\newcommand{\QCCcyc}{\ensuremath{\ell}}
\newcommand{\QCCext}{\ensuremath{r}}

\DeclareDocumentCommand \eigenvector { ooo }
{
	\IfNoValueTF {#3}
	{  
		\IfNoValueTF {#2}
			{ 
				\IfNoValueTF {#1}
				{ 	
					\ensuremath{\mathbf{v}}	
				}
				{ 
					\ensuremath{v_{#1}}	
				}
			}
			{ 
			\ensuremath{\mathbf{v}_{#1}^{\langle #2 \rangle}}			
			}
	}
	{
	\ensuremath{{v}_{#1,#3}^{\langle #2 \rangle}}
	}
}
\DeclareDocumentCommand \supp{o}
{
	\IfNoValueTF {#1}
	{  
		\ensuremath{\mathcal{Y}}
	}
	{
		\ensuremath{\mathcal{Y}_{#1}}
	}
}
\DeclareDocumentCommand \errorsup{o}
{
	\IfNoValueTF {#1}
	{  
		\ensuremath{\mathcal{E}}
	}
	{
		\ensuremath{\mathcal{E}_{#1}}
	}
}
\DeclareDocumentCommand \noerrors{o}
{
	\IfNoValueTF {#1}
	{  
		\ensuremath{\mathcal{\varepsilon}}
	}
	{
		\ensuremath{\mathcal{\varepsilon}_{#1}}
	}
}
\DeclareDocumentCommand \eigenspace{o}
{
	\IfNoValueTF {#1}
	{  
		\ensuremath{\mathcal{V}}
	}
	{
		\ensuremath{\mathcal{V}_{#1}}
	}
}
\DeclareDocumentCommand \gen { ooo }
{
	\IfNoValueTF {#3}
	{  
		\IfNoValueTF {#2}
			{ 
				\IfNoValueTF {#1}
				{ 	
					\ensuremath{g(X)}	
				}
				{ 
					\ensuremath{g^{#1}(X)}	
				}
			}
			{ 
			\ensuremath{g_{#1,#2}(X)}			
			}
	}
	{
	\ensuremath{g_{#2,#3}^{#1}(X)}
	}
}
\DeclareDocumentCommand \genbar { ooo }
{
	\IfNoValueTF {#3}
	{  
		\IfNoValueTF {#2}
			{ 
				\IfNoValueTF {#1}
				{ 	
					\ensuremath{\bar{g}(X)}	
				}
				{ 
					\ensuremath{\bar{g}^{#1}(X)}	
				}
			}
			{ 
              \ensuremath{\bar{g}_{#1,#2}(X)}			
			}
	}
	{
      \ensuremath{\bar{g}_{#2,#3}^{#1}(X)}
	}
}
\DeclareDocumentCommand \genarg { moom }
{
  \IfNoValueTF {#2}
  {
    \ensuremath{g^{#1}(#4)}   
  }
  { 
    \ensuremath{g^{#1}_{#2,#3}(#4)}
  }
}
\DeclareDocumentCommand \genmat { oo }
{
  \IfNoValueTF {#2}
  { 
    \IfNoValueTF {#1}
    {   
      \ensuremath{\mathbf{G}(X)}
    }
    { 
      \ensuremath{\mathbf{G}^{#1}(X)}    
    }
  }
  { 
    \ensuremath{\mathbf{G}^{#1}(#2)}           
  }
}

\newcommand{\interval}[1]{\ensuremath{[#1)}}
\newcommand\bigzero{\makebox(0,0){\text{\huge0}}}
\newcommand\bigzerob{\hspace{-10ex}\smash{\clap{\resizebox{0.35cm}{!}{$0$}}}}

\newcommand{\diag}{\mathrm{diag}} 
\newcommand{\coset}[2]{\ensuremath{M_{#2}^{\langle #1 \rangle}}}
\newcommand{\minpoly}[2]{\ensuremath{m_{#2}^{\langle #1 \rangle}(X)}}

\newcommand{\rowop}[1]{\ensuremath{\mathsf{R}[#1]}}
\newcommand{\level}{\ensuremath{r}}
\newcommand{\basis}{\ensuremath{\tilde{K}}}
\newcommand{\map}[2]{\ensuremath{\mu(#1,#2)}}
\newcommand{\mapb}[2]{\ensuremath{\bar{\mu}(#1,#2)}}
\newcommand\minus{%
  \setbox0=\hbox{-}%
  \vcenter{%
    \hrule width\wd0 height \the\fontdimen8\textfont3%
  }%
}